\documentclass[
submission
]{dmtcs-episciences}

\usepackage{amsmath}
\usepackage[utf8]{inputenc}
\usepackage{subfigure}

\usepackage[round]{natbib}
\usepackage[ruled,vlined]{algorithm2e}

\newtheorem{definition}{Definition}[section]
\newtheorem{theorem}{Theorem}[section]
\newtheorem{lemma}{Lemma}[section]

\newtheorem{cor}{Corollary}
\newtheorem{obs}{Observation}[section]
\newtheorem{claim}{Claim}[section]

\author{Michael A. Henning\affiliationmark{1}
  \and Arti Pandey\affiliationmark{2}
  \and Vikash Tripathi\affiliationmark{2}}
\title[Semipaired Domination in Some Subclasses of Chordal Graphs]{Semipaired Domination in Some Subclasses of Chordal Graphs}
\affiliation{
Department of Mathematics and Applied Mathematics, University of Johannesburg Auckland Park, South Africa\\
Department of Mathematics, Indian Institute of Technology Ropar, Nangal Road, Rupnagar, Punjab, INDIA}
\received{2020-09-15}
\revised{2021-03-15}
\accepted{2021-06-14}

\begin{document}
\publicationdetails{23}{2021}{1}{19}{6782}
\maketitle
\begin{abstract}
 A dominating set $D$ of a graph $G$ without isolated vertices is called semipaired dominating set if $D$ can be partitioned into $2$-element subsets such that the vertices in each set are at distance at most $2$. The semipaired domination number, denoted by $\gamma_{pr2}(G)$ is the minimum cardinality of a semipaired dominating set of $G$. Given a graph $G$ with no isolated vertices, the \textsc{Minimum Semipaired Domination} problem is to find a semipaired dominating set of $G$ of cardinality $\gamma_{pr2}(G)$. The decision version of the \textsc{Minimum Semipaired Domination} problem is already known to be NP-complete for chordal graphs, an important graph class. In this paper, we show that the decision version of the \textsc{Minimum Semipaired Domination} problem remains NP-complete for split graphs, a subclass of chordal graphs. On the positive side, we propose a linear-time algorithm to compute a minimum cardinality semipaired dominating set of block graphs. In addition, we prove that the \textsc{Minimum Semipaired Domination} problem is APX-complete for  graphs with maximum degree $3$.
\end{abstract}

\section{Introduction}
\label{sec:1}

For a graph $G=(V,E)$, a vertex $v\in V$ is said to \emph{dominate} a vertex $w\in V$ if either $v=w$ or $vw\in E$. A \textit{dominating set} of $G$ is a set $D \subseteq V$ such that every vertex in $V$ is dominated by at least one vertex of $D$. The minimum cardinality of a dominating set of $G$ is called the \textit{domination number} of $G$ denoted by $\gamma(G)$. Many facility location problems can be modelled using the concept of domination in graphs. Due to vast applications of domination, many variations of dominations have also been introduced in the literature. The domination problem are widely studied from combinatorial as well as algorithmic point of view, see~\cite{hhs1,hhs2}.

One important variation of domination is paired domination. The concept of paired domination was introduced by Haynes and Slater in \cite{paired}. For a graph $G$ with no isolated vertices, a dominating set $D$ is called a \emph{paired dominating set}, abbreviated as PD-set, if the graph induced by $D$ has a perfect matching $M$. Two vertices joined by an edge of $M$ are said to be \emph{paired} and are also called \emph{partners} in $D$. The \textsc{Minimum paired domination} problem is to find a PD-set of $G$ of minimum cardinality. The cardinality of such a set is known as the \emph{paired domination number} of $G$, and is denoted by $\gamma_{pr}(G)$. A survey of paired domination can be found in \cite{paired2}.

A relaxed version of paired domination, known as \textit{semipaired domination} was introduced by Haynes and Henning~\cite{semi-paired3}, and further studied by others~\cite{semi-paired4,semi-paired2,semi-paired5,iwoca,semi-paired6}. For a graph $G$ with no isolated vertex, a \textit{semipaired dominating set}, abbreviated as semi-PD-set, is a dominating set $D$ of $G$ such that the vertices in $D$ can be partitioned into $2$-sets such that if $\{u,v\}$ is a $2$-set, then $uv \in E(G)$ or the distance between $u$ and $v$ is~$2$. We say that $u$ and $v$ are \emph{semipaired}, and that $u$ and $v$ are \emph{partners}. The minimum cardinality of a semi-PD-set of $G$ is called the \textit{semipaired domination number} of $G$, and is denoted by $\gamma_{pr2}(G)$. For a graph $G$ with no isolated vertices, the \textsc{Minimum Semipaired Domination} problem is to find a semi-PD-set of cardinality $\gamma_{pr2}(G)$. For a  given graph $G$ and a positive integer $k$, the \textsc{Semipaired Domination Decision} problem is to determine whether $G$ has a semi-PD-set of cardinality at most $k$ or not. Since every PD-set is a semi-PD-set, and since every semi-PD-set is a dominating set, we have the following observation.

\begin{obs}{  (\cite{semi-paired3})}
 \label{ob:chain}
For every graph $G$ without isolated vertices,
\[
\gamma(G) \le \gamma_{pr2}(G) \le \gamma_{pr}(G).
\]
\end{obs}

The algorithmic study of the \textsc{Minimum Semipaired Domination} problem was initiated by Henning et al. in \cite{iwoca}. They
 proved that  the \textsc{Semipaired Domination Decision} problem is NP-complete even for bipartite graphs and chordal graphs. They also proposed a linear-time algorithm to compute a minimum cardinality semi-PD-set of an interval graph. They proposed a $1+\ln(2\Delta+2)$-approximation algorithm for the \textsc{Minimum Semipaired Domination} problem, where $\Delta$ denotes the maximum degree of the graph. On the negative side, they proved that \textsc{Minimum Semipaired Domination} problem cannot be approximated within $(1-\epsilon) \ln|V|$ for any
$\epsilon > 0$ unless P$=$NP. In this paper, we continue the algorithmic study of the \textsc{Minimum Semipaired Domination} problem. The main contributions of the paper are summarized below.

In Section~\ref{sec:2}, we discuss some definitions and notations. In this section we also observe some graph classes where paired domination and semipaired domination problems differ in complexity. In Section $3$, we prove the \textsc{Semipaired Domination Decision} problem is NP-complete for split graphs. In Section $4$, we propose a linear time algorithm to compute a minimum cardinality semipaired dominating set in block graphs. In Section 5, we show that the \textsc{Minimum Semipaired Domination} problem is APX-hard for graphs with maximum degree $3$.  Finally, Section $6$ concludes the paper.
\vspace*{.5cm}
\section{Preliminaries}
\label{sec:2}
\vspace*{.2cm}
\subsection{Definitions and Notations}
Let $G =(V, E)$ be a  graph. For a vertex $v\in V$, let $N_G(v)=\{u\in V \mid uv\in E\}$ and $N_G[v]=N_G(v)\cup \{v\}$ denote the {\it open neighborhood} and the {\it closed neighborhood} of $v$, respectively. For a set $S\subseteq V$, the sets $N_{G}(S)=\bigcup_{u \in S} N_{G}(u)$ and $N_{G}[S]=N_{G}(S)\cup S$ are called {\it open neighborhood} and the {\it closed neighborhood} of $S$, respectively.
 For a set $S \subset V$, the graph $G\setminus S$ is obtained from $G$ by deleting all vertices in $S$ and all edges incident with vertices in $S$. If $S = \{v\}$, we write $G \setminus v$ rather than $G \setminus\{v\}$.
A \emph{cut vertex} in a connected graph $G$ is a vertex $v \in V$ such that $G \setminus v$ is disconnected. Let $n$ and $m$ denote the number of vertices and edges of $G$, respectively. We use standard notation, $[k] =\{1,2, \ldots, k \}$. In this paper, we only consider connected graphs with at least two vertices.

A set $S\subseteq V$ is called an \emph{independent
set} of $G$ if $uv\notin E$ for all $u,v\in S$. A set
$K\subseteq V$ is called a \emph{clique} of $G$ if $uv\in E$ for all $u,v\in K$. A graph
$G$ is said to be a \emph{chordal graph} if every cycle in $G$ of length at least four has a \emph{chord}, that is, an edge joining two non-consecutive vertices of the cycle. A chordal graph $G$ is a \emph{split graph} if $V$ can be partitioned into two sets $C$ and $I$ such that $C$ is a clique and $I$ is an independent set.

A \emph{rooted tree} $T$ distinguishes one vertex $r$ called the \emph{root}. For each vertex $v \ne r$ of $T$, the \emph{parent} of $v$ is the neighbor of $v$ on the unique $(r,v)$-path, while a
\emph{child} of $v$ is any other neighbor of $v$. Further, the \emph{grandparent} of $v$ is the vertex at distance~$2$ from $v$ on the unique $(r,v)$-path. A \emph{descendant} of $v$ is a vertex $u \ne v$ such that the unique $(r,u)$-path contains $v$. A
\emph{grandchild} of $v$ in $T$ is a descendant of $v$ at distance~$2$ from $v$.
\vspace*{.2cm}
\subsection{Complexity difference between paired domination and semipaired domination}

In this subsection, we discuss the complexity difference between paired domination and semipaired domination. We show that for the class of GP$4$ graphs, which we define below, the decision version of the \textsc{Minimum Paired Domination} problem is NP-complete, but the \textsc{Minimum Semipaired Domination} problem is easily solvable. On the other hand, we introduce a graph class called GP$5$ graphs, and we show that the \textsc{Semipaired Domination Decision} problem is NP-complete for GP$5$ graphs, but the \textsc{Minimum Paired Domination} problem is easily solvable for this graph class.

The class of GP$4$ graphs was introduced by Henning et al. in \cite{semito}. Below we recall the definition of GP$4$ graphs.

\begin{definition}[GP$4$-graph]
\label{defn1}
{\rm
A graph $G=(V,E)$ is called a \emph{GP}$4$-\emph{graph} if it can be obtained from a general connected graph $H=(V_{H},E_{H})$ where $V_{H}=\{v_{1},v_{2},\ldots,v_{n_{H}}\}$, by adding a path of length~$4$ to every vertex of $H$. Formally, $V = V_{H} \cup \{ w_{i},x_{i},y_{i},z_{i} \mid  i \in [n_{H}] \, \}$ and $E=E_{H}\cup \{v_{i}w_{i},w_{i}x_{i},x_{i}y_{i},y_{i}z_{i}\mid  i \in [n_H] \, \}$, where $n_{H}$ denotes the number of vertices in $H$.
}
\end{definition}

\begin{theorem}
 \label{t:GP4}
If $G$ is a \emph{GP}$4$-\emph{graph}, then $\gamma_{pr2}(G) = \frac{2}{5}|V(G)|$.
\end{theorem}

\begin{proof}
Let $G$ be a GP$4$-graph of order $n=|V(G)| = 5|V_{H}|$ as constructed in Definition~\ref{defn1}. The set $S = \{w_i, y_i \mid i \in [n_H]\}$ is a semi-PD-set of $G$, implying that $\gamma_{pr2}(G) \leq \frac{2}{5}|V(G)|$. Every semi-PD-set of $G$ must contain at least two vertices from the set $\{ w_i, x_i ,y_i ,z_i\}$ for each $i \in [n_{H}]$. Thus, $\gamma_{pr2}(G) \geq \frac{2}{5}|V(G)|$.  Consequently, $\gamma_{pr2}(G) = \frac{2}{5}|V(G)|$.
\end{proof}

\begin{lemma}
\label{l:lemGP4}
If $G$ is a \emph{GP}$4$-\emph{graph} constructed from a graph $H$ as in Definition~\ref{defn1}, then $H$ has a PD-set of cardinality at most $k$, $k\leq n_H$, if and only if $G$ has a PD-set of cardinality at most $2 n_H +k$.
\end{lemma}

\begin{proof} Suppose $D$ is a PD-set of $H$ of cardinality at most~$k$. Then the set $D\cup \{x_i,y_i \mid i \in [n_H]\}$ is a PD-set of $G$ of cardinality at most $2 n_H +k$. Conversely, assume that $G$ has a PD-set $D'$ of cardinality at most~$2 n_H +k$. Now we obtain a PD-set of $H$ of size at most $k$ by updating $D'$. In order to dominate $z_i$, the set $D'$ must contain either $y_i$ or $z_i$. Without loss of generality, we may assume that $D'$ contains $y_{i}$ and that $x_{i} \in D'$ with $x_{i}$ and $y_i$ paired in $D'$. This implies that $D'$ contains exactly two vertices from the set $\{x_i,y_i,z_i\}$, where $i \in [n_H]$. Now we consider the set $D = D' \setminus \{x_i,y_i,z_i \mid i \in [n_H] \}$.  We note that $|D| \leq k$ and $D\subseteq \{v_{i},w_{i}\mid i \in [n_H]\}$. Further, the set $D$ dominates $V_{H}$, and for every vertex $u$ in $D$, the partner of $u$ is also present in $D$. Let $G'=G[V(H) \cup \{w_i \mid i \in [n_H]\}]$. Observe that $w_i$ can be paired only with $v_i$ in $G'$ for $i \in [n_H]$. If $w_i \in D$ and $N(v_i) \subseteq D$, then we update the set $D$ as $D \setminus \{w_i,v_i\}$. If there exist a vertex $u \in N(v_i)$ such that $u \notin D$, then we update the set $D$ as $D = (D \setminus \{w_i\}) \cup \{u\}$. We do this update for each $w_i \in D$. Note that the updated set $D$ is a PD-set of $H$ and $|D| \leq k$. Hence, the result follows.
\end{proof}

Since the decision version of the \textsc{Minimum Paired Domination} problem is known to be NP-complete for general graphs~\cite{paired}, the following theorem follows directly from Lemma~\ref{l:lemGP4}.

\begin{theorem}
The decision version of the \textsc{Minimum Paired Domination} problem is NP-complete for \emph{GP}$4$-\emph{graphs}.
\end{theorem}

Next, we define a new graph class and call it GP$5$-graphs.
 \\ \vspace*{.4cm}
\begin{definition}[GP$5$-graph]
\label{defn2}
{\rm
A graph $G=(V,E)$ is called a \emph{GP}$5$-\emph{graph} if it can be obtained from a general connected graph $H=(V_{H},E_{H})$ where $V_{H}=\{v_{1},v_{2},\ldots,v_{n}\}$, by adding a vertex disjoint path $P_5$ for each vertex $v$ of $H$ and joining $v$ to the central vertex of the path. Formally, $V = V_{H} \cup \{ a_{i},b_{i},c_{i},d_{i},e_{i} \mid i \in [n] \}$ and $E=E_{H}\cup \{v_{i}c_{i},c_{i}b_{i},c_{i}d_{i},b_{i}a_{i},d_{i}e_{i}\mid i \in [n]  \}$, where $n$ denotes the number of vertices in $H$.
}
\end{definition}

For example, when $H$ is a $4$-cycle $C_4$, then a GP$5$ graph obtained from $H$ is shown in Fig.~\ref{fig:gp5}. For a GP$5$ graph, we show that the \textsc{Semipaired Domination Decision} problem is NP-complete, but the \textsc{Minimum Paired Domination} problem is easily solvable.

\begin{figure}[htbp]
  \begin{center}
    \includegraphics[width=0.9\textwidth]{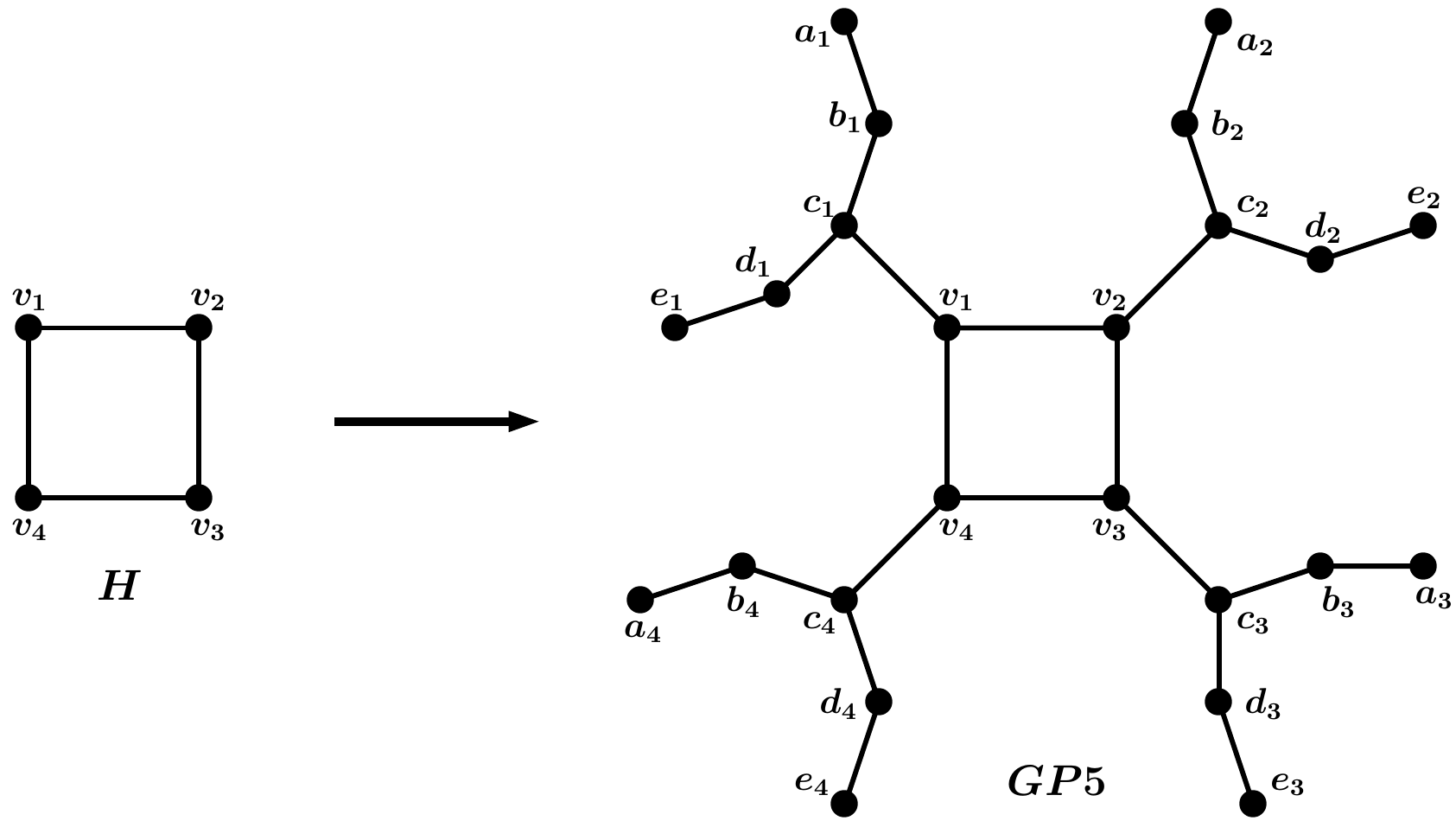}
    \caption{An illustration of a GP$5$ graph obtained from a $4$-cycle.}
    \label{fig:gp5}
  \end{center}
\end{figure}

\begin{theorem}
 \label{t:GP5}
If $G$ is a \emph{GP}$5$-\emph{graph}, then $\gamma_{pr}(G) = \frac{2}{3}|V(G)|$.
\end{theorem}
\begin{proof}
 Let $G$ be a GP$5$-graph of order $n=|V(G)|=6|V_H|$ as constructed in Definition~\ref{defn2}. The set $S=\{a_i, b_i,c_i,d_i \mid i \in [n_H]\}$ is a PD-set of $G$, implying that $\gamma_{pr}(G) \leq \frac{2}{3}|V(G)|$. Also we note that every PD-set of $G$ must contain at least four vertices from the set $\{a_i,b_i,c_i,d_i,e_i\}$ for each $i \in [n_H]$. Hence, $\gamma_{pr}(G) \geq \frac{2}{3}|V(G)|$. Consequently, $\gamma_{pr}(G) = \frac{2}{3}|V(G)|$.
\end{proof}

\begin{lemma}
\label{l:lemGP5}
If $G$ is a \emph{GP}$5$-\emph{graph} constructed from a graph $H$ as in Definition~\ref{defn2}, then $H$ has a semi-PD-set of cardinality~$k$, $k\leq n_H$ if and only if $G$ has a semi-PD-set of cardinality $2 n_H +k$.
\end{lemma}
\begin{proof}
Suppose $D$ is semi-PD-set of $H$ of cardinality at most~$k$. Then the set $D\cup \{b_i,d_i \mid i \in [n_H]\}$ is a semi-PD-set of $G$ of cardinality at most $2 n_H +k$. Conversely, assume that $G$ has a semi-PD-set $D'$ of cardinality at most~$2 n_H +k$. Now we obtain a PD-set of $H$ of size at most $k$ by updating $D'$. In order to dominate $a_i$, the set $D'$ must contain either $a_i$ or $b_i$. Similarly, in order to dominate $e_i$, the set $D'$ must contain either $d_i$ or $e_i$. Without loss of generality, for each $i \in [n_H]$ we may assume that $D'$ contains $\{b_{i},d_{i}\}$ with $b_{i}$ and $d_i$ semipaired in $D'$. This implies that $D'$ contains exactly two vertices from the set $\{a_i,b_i,d_i,e_i\}$, where $i \in [n_H]$. Now we consider the set $D = D' \setminus \{a_i,b_i,d_i,e_i \mid i \in [n_H] \}$.  We note that $|D| \leq k$ and $D\subseteq \{v_{i},c_{i}\mid i \in [n_H]\}$. Further, the set $D$ dominates $V_{H}$, and for every vertex $u$ in $D$, the semipair of $u$ is also present in $D$. Let $G'=G[V(H) \cup \{c_i \mid i \in [n_H]\}]$. If $c_i \in D$ and $c_i$ is semipaired with a vertex $v_i$ such that $N(v_i) \subseteq D$, then we update the set $D$ as $D \setminus \{c_i,v_i\}$ and if there exist a vertex $u \in N(v_i)$ such that $u \notin D$, then we update the set $D$ as $D = (D \setminus \{c_i\}) \cup \{u\}$. We do this update for each $c_i \in D$. Note that the updated set $D$ is a semi-PD-set of $H$ and $|D| \leq k$. Hence, the result follows.
\end{proof}

Since the  \textsc{Semipaired Domination Decision} problem is known to be NP-complete for general graphs~\cite{iwoca}, the following theorem follows directly from Lemma~\ref{l:lemGP5}.

\begin{theorem}
The \textsc{Semiaired Domination Decision} problem is NP-complete for \emph{GP}$4$-\emph{graphs}.
\end{theorem}
 \vspace*{.5cm}
\section{NP-Completeness Result for Split Graphs}
\label{Sec:3}

In this section, we prove that the \textsc{Semipaired Domination Decision} problem is NP-complete for split graphs. To prove this NP-completeness result, we use a reduction from  the domination problem, which is a well known NP-complete problem~\cite{hhs1}.

\begin{theorem}\label{t:split}
The \textsc{Semipaired Domination Decision} problem is NP-complete for split graphs.
\end{theorem}

\begin{proof}
Clearly, the \textsc{Semipaired Domination Decision} problem is in NP. To show the hardness, we give a polynomial time reduction from the \textsc{Domination Decision} problem for general graphs. Given a non-trivial graph $G=(V,E)$, where $V=\{v_{i}\mid i\in [n]\}$ and $E=\{e_{j}\mid j\in [m]\}$, we construct a split graph $G'=(V_{G'},E_{G'})$ as follows:

Let $V_{k}=\{v_{i}^{k} \mid i \in [n]\}$ and $U_{k}=\{u_{i}^{k}\mid  i \in [n]\}$ for $k\in [2]$.
Now define $V_{G'}=V_{1}\cup V_{2}\cup U_{1}\cup U_{2}$, and $E_{G'}=\{uv \mid u,v \in V_1 \cup U_1, u\neq v\} \cup \{v_i^2v_j^1,u_i^2u_j^1 \mid  i \in [n]$ and $v_j \in N_G[v_i]\}$. Note that the set $A=V_1 \cup U_1$ is a clique in $G'$ and the set $B=V_2 \cup U_2$ is an independent set in $G'$. Since $V_{G'}=A \cup B$, the constructed graph $G'$ is a split graph. Fig.~\ref{fig:split} illustrates the construction of $G'$ from $G$.
\vspace*{.5 cm}

\begin{figure}[htbp]
  \begin{center}
    \includegraphics[width=0.9\textwidth]{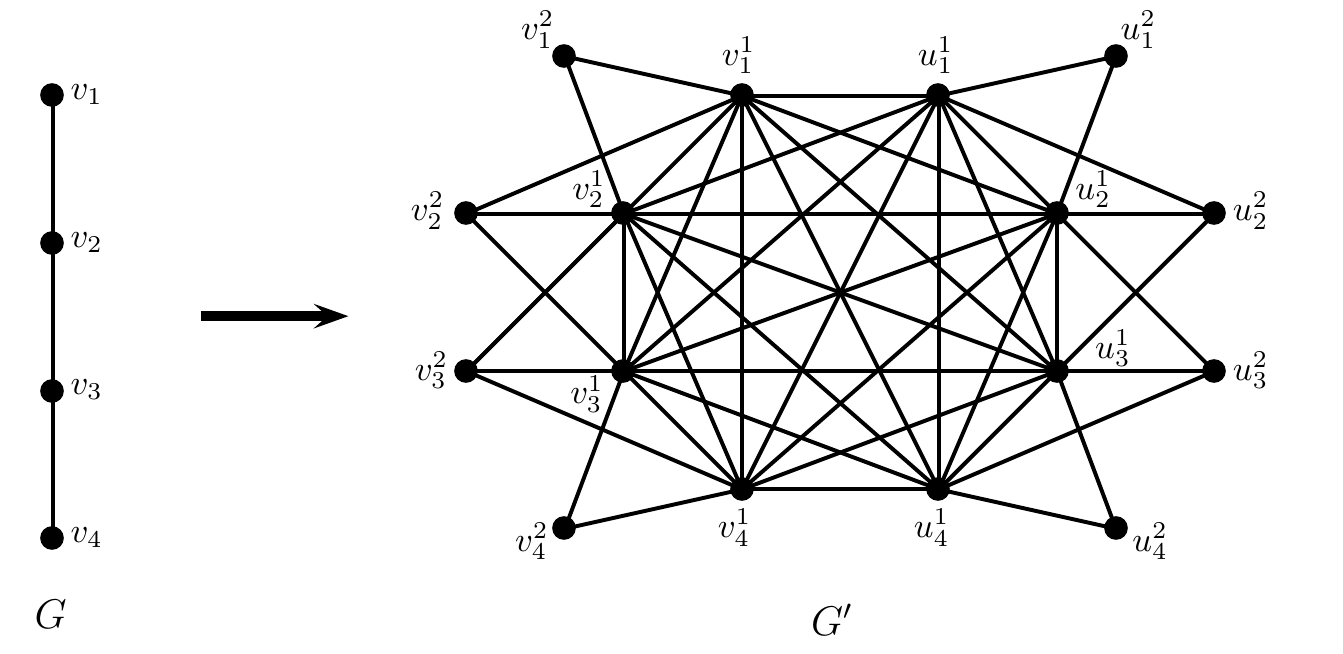}
    \caption{An illustration to the construction of $G'$ from $G$ in the proof of Theorem~\ref{t:split}}
    \label{fig:split}
  \end{center}
\end{figure}

Now, to complete the proof of the theorem, we only need to prove the following claim.

\begin{claim}
\label{c:claim1}
The graph $G$ has a dominating set of cardinality at most $k$ if and only if $G'$ has a semi-PD-set of cardinality at most $2k$.
\end{claim}

\noindent
\begin{proof}
Let $D = \{v_{i_1}, v_{i_2}, \ldots ,v_{i_k}\}$ be a dominating set of $G$ of cardinality at most~$k$. Then the set $D' = \{v_{i_1}^1, v_{i_2}^1, \ldots ,v_{i_k}^1\} \cup \{u_{i_1}^1, u_{i_2}^1, \ldots ,u_{i_k}^1\}$ is a semi-PD-set of $G'$ of cardinality at most $2k$.

Conversely, let $D'$ is a semi-PD-set of $G'$ of cardinality at most $2k$. Now we obtain a dominating set of $G$ of size at most $k$ by updating $D'$. Note that, either $|D' \cap(V_1 \cup V_2)| \leq k$ or $|D' \cap(U_1 \cup U_2)| \leq k$. Without loss of generality, we assume that $|D' \cap(V_1 \cup V_2)| \leq k$. Let $D = D' \cap(V_1 \cup V_2)$. Note that $V_2$ is an independent set, hence any vertex $v \in V_2$ is either dominated by itself or by some vertex in $V_1$. If $v_i^2 \in D$ and none of its neighbors is in $D$, then update $D = (D \setminus \{v_i^2\}) \cup \{u\}$ where $u \in N(v_i^2)$. We do this update for each vertex $v_i^2 \in V_2$. Now observe that in the updated set $D$, we have $N(v_i^2) \cap D \neq \emptyset $ for $i \in [n]$. The set $D'' = \{v_i \mid v_i^1 \in D\}$ is a dominating set of $G$ of cardinality at most~$k$. This proves the claim.
\end{proof}

Hence the result follows.
\end{proof}

\section{Semipaired Domination in Block Graphs}
\label{Sec:4}

 For a graph $G$, a maximal induced subgraph of $G$ without a cut vertex is called a \emph{block} of $G$. If $B$ and $B'$ are two blocks of $G$ then $|V(B) \cap V(B')| \leq 1$, and a vertex $v \in V(B) \cap V(B')$ if and only if $v$ is a cut vertex. A connected graph whose every block is a complete graph is called a \emph{block graph}. A tree is a block graph in which every block contains exactly two vertices. A block with only one cut vertex is called an \emph{end block}. Every block graph not isomorphic to a complete graph has at least two end blocks.

Lie Chen et al.~\cite{paired1} have studied an ordering of vertices of block graph, $\alpha =(v_1, v_2, \ldots , v_n)$ such that $v_iv_j \in E$ and $v_iv_k \in E$ implies $v_jv_k \in E$ for $i<j<k\leq n$. Such an ordering of vertices of a block graph is called \emph{Block-Elimination-Ordering$ \, ($BEO$)$}. The procedure to get such an ordering is as follows: if $G$ is not isomorphic to a complete graph, then it must have at least two end blocks. Pick an end block, say $B$, having a cut vertex $x$. Staring with the index $1$, enumerate the vertices in $V(B) \setminus \{x\}$ in any order and remove $V(B)\setminus \{x\}$ from the graph. Let $k = \max \{s \mid v_s \in V(B)\setminus \{x\}\}$, that is, $v_k$ is the vertex in $V(B)\setminus \{x\}$ having highest index. Now if the remaining graph, say $G'$, is a complete graph, then enumerate the remaining vertices starting from index $k+1$ to $n$ in any order; otherwise, pick an end block in $G'$, say $B'$, having cut vertex $x'$. Starting with index $k+1$ enumerate the vertices in $V(B') \setminus \{x'\}$ and continue the procedure in $G' \setminus (V(B') \setminus \{x'\})$.

 Let  $G=(V,E)$ be a block graph, and  $\alpha= (v_1, v_2, \ldots ,v_n)$ be a BEO of vertices of $G$. For $i\neq n$, we define $F(v_i)=v_j$, where $j = \max \{k \mid v_iv_k \in E\}$. We also define $F(v_{n})=v_n$. Further, we construct a block tree $T(G)$ rooted at $v_n$ such that $V(T(G))=V(G)$ and $E(T(G))=\{uv$ if and only if either $F(u)=v$ or $F(v)=u\}$. Fig.~\ref{fig:mix} illustrates the construction of $T(G)$ from a block graph $G$. Note that a cut vertex of $G$ is an internal vertex of $T(G)$. Also if $B$ is a block of $G$ with $V(B) = \{u_{i_1}, u_{i_2}, \ldots, u_{i_k}\}$ where $u_{i_k}$ is the highest index vertex in $V(B)$, then $u_{i_1}, u_{i_2}, \ldots, u_{i_{k-1}}$ are called \emph{siblings}  in $T(G)$ and each one is a child of $u_{i_k}$. The following complexity is already known.

 \begin{theorem}\label{th:6}
 \cite{paired1}~For a block graph $G=(V,E)$, a BEO can be computed in O$(n+m)$-time. In addition, given a BEO, the corresponding block tree can also be computed in O$(n+m)$-time.
 \end{theorem}
\vspace*{.8cm}
 \begin{figure}[htbp]
  \begin{center}
    \includegraphics[width=1 \textwidth]{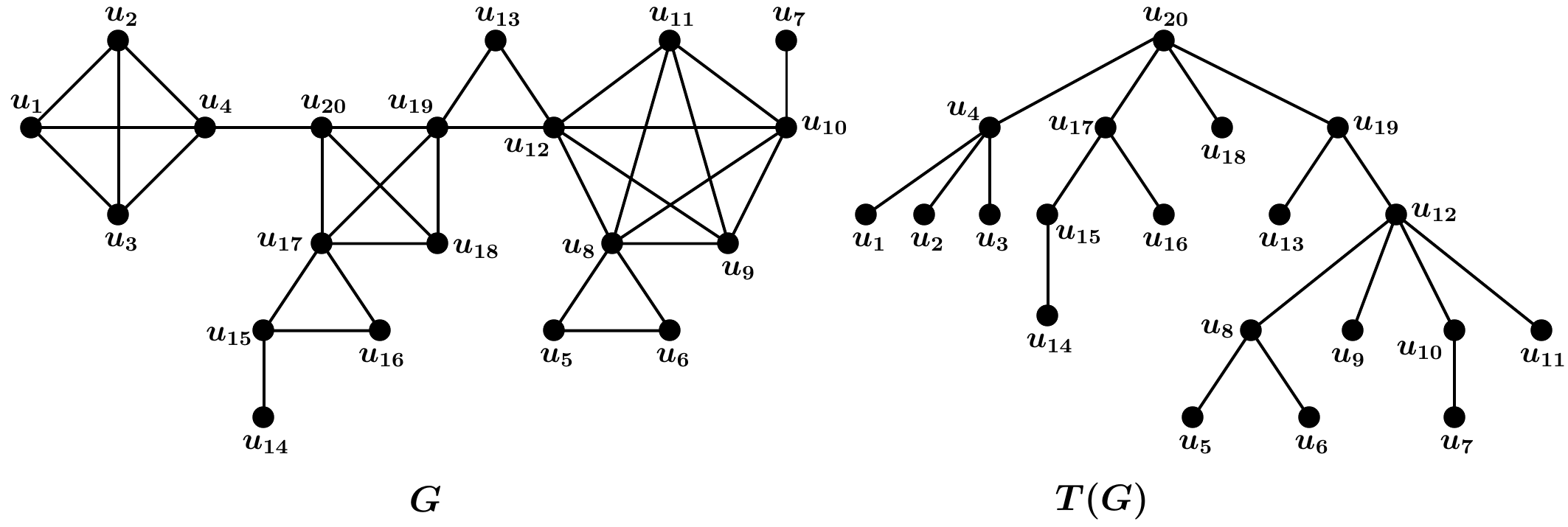}
    \caption{An illustration of the construction of $T(G)$ from a block graph $G$.}
    \label{fig:mix}
  \end{center}
\end{figure}

\begin{obs}
 \label{ob:block1}
Let $G=(V(G), E(G))$ be a block graph and let $T(G)$ be a corresponding block tree. If $v \in N_G(u)$, then one of the following holds in $T(G)$,
\begin{enumerate}
\item[{\rm (a)}] $u$ is a parent of $v$.
\item[{\rm (b)}] $u$ is a child of $v$.
\item[{\rm (c)}] $u$ is a sibling of $v$.
\end{enumerate}
\end{obs}

\begin{lemma}
\label{l:lemma3}
Let $G$ be a block graph with given BEO, $\alpha =( v_1, v_2, \ldots ,v_n)$. If $B$ and $B'$ are any two blocks of $G$ such that $v_i \in V(B) \cap V(B')$, then $F(u)=v_i$, for all $u \in V(B) \setminus \{v_i\}$ or for all $u \in V(B') \setminus \{v_i\}$.
\end{lemma}

\noindent
\begin{proof}
Let $G$ be a block graph with given BEO, $\alpha =( v_1, v_2, \ldots ,v_n)$. Suppose $B$ and $B'$ are blocks of $G$ such that $v_i \in V(B) \cap V(B')$. Clearly, $v_i$ is a cut vertex. By the way the vertices of $G$ are enumerated, either all the vertices in $V(B) \setminus \{v_i\}$ first get enumerated and thereafter the vertices in $V(B')$, or all the vertices in $V(B') \setminus \{v_i\}$ first get enumerated and then the vertices in $V(B)$. Renaming the blocks if necessary, we may assume without loss of generality that the vertices in $V(B) \setminus \{v_i\}$ get enumerated first and thereafter the vertices in $V(B')$. In that case, we note that $i \geq \max \{r \mid v_r \in V(B)\setminus \{v_{i}\}\}$. Thus, $F(u)=v_i$ for all $u \in V(B) \setminus \{v_i\}$, implying the desired result.
\end{proof}

Using above lemma, we state that if there are exactly $k$ blocks, say $B_1, B_2 \ldots B_k$ of a block graph $G$ such that $v_i \in V(B_1) \cap V(B_2) \cap \cdots \cap V(B_k)$ and $B^* \in \{B_1, B_2 \ldots B_k\}$ is the block whose vertices are enumerated after the enumeration of vertices in $\{V(B_1) \cup V(B_2) \cup \cdots \cup V(B_k)\} \setminus \{V(B^*)\}$ in the BEO $\alpha$ then for every vertex $v_j \in \{V(B_1) \cup V(B_2) \cup \cdots \cup V(B_k)\} \setminus \{V(B^*)\}$, $F(v_j) = v_i$.

\begin{lemma}
\label{l:lem4}
Let $G$ be a block graph and let $T(G)$ be a block tree of $G$. If $\alpha=(v_n,v_{n-1}, \ldots, v_1)$ is a BFS-ordering of the vertices of $T(G)$ rooted at $v_n$ then the reverse of BFS-ordering $\beta=(v_1,v_{2}, \ldots, v_{n})$ also satisfy BEO in $G$.
\end{lemma}
\begin{proof}
Let $G$ be a block graph and $T(G)$ be the corresponding block tree. Let  $\beta=(v_1,v_{2}, \ldots, v_{n})$ be the reverse of BFS-ordering $\alpha=(v_n,v_{n-1}, \ldots, v_1)$ of vertices of $G$ as they appear in $T(G)$. For $i < j <k$, let $v_i$, $v_j$ and $v_k$ satisfy the reverse of BFS-ordering and $v_iv_j, v_iv_k \in E(G)$. To prove the result we need to show $v_jv_k \in E(G)$. By contradiction, suppose $v_jv_k \notin E(G)$. Since $v_jv_k \notin E(G)$, this implies $v_j$ and $v_k$ belongs to different blocks of $G$, say $B$ and $B'$ respectively, and $v_i$ is a cut vertex. Now, using Lemma~\ref{l:lemma3} we have $F(u)=v_i$, for all $u \in V(B) \setminus \{v_i\}$ or for all $u \in V(B') \setminus \{v_i\}$. Therefore in $T(G)$, the vertex $v_i$ is the parent of either $v_j$ or $v_k$, which is a contradiction as $v_i$, $v_j$ and $v_k$ satisfy the reverse of BFS-ordering and $i < j < k$. Hence the result follows.
\end{proof}

Let $G=(V(G),E(G))$ be a block graph and $T(G)$ be a corresponding block tree of $G$. Let $\alpha=(v_n,v_{n-1}, \ldots, v_1)$ be a $BFS$-ordering of vertices of $T(G)$ rooted at $v_n$. Recall that the distance between two vertices $u$ and $v$ is the length of shortest path between $u$ and $v$, denote by $d(u,v)$. For a positive integer $l \geq 0$, we say a vertex $x$ is at level $l$ in a tree $T$ rooted at a vertex $y$, if $d_{T}(x,y) =l$. In our algorithm, we will process the vertices of the block graph as they appear in the reverse of $BFS$-ordering $\beta=$ $(v_1, v_{2}, \ldots , v_n)$ of the corresponding block tree $T(G)$. We will use the following notation while processing the vertices in the algorithm:

\[
  D(v_i) =
  \begin{cases}
     0 & \text{if $v_i$ is not dominated}, \\
     1 & \text{if $v_i$ is dominated}. \\
  \end{cases}
\]

 \[
  L(v_i) =
  \begin{cases}
     0 & \text{if $v_i$ is not selected}, \\
     1 & \text{if $v_i$ is selected but not semipaired}, \\
     2 & \text{if $v_i$ is selected and semipaired}.\\
  \end{cases}
\]

 \[  m(v_i) =
  \begin{cases}
     k & {\text{if $v_k$ need to be semipaired with a vertex in $N_{T(G)}[v_i]$ or with some sibling of $v_i$}}, \\
     0 & \text{otherwise}. \\
  \end{cases}
\]

Also we use $N_i(v_k)=\{v_j \mid v_kv_j \in E(G)$ and $ j \geq i \}$ and $N_i[v_k]=\{v_j \mid v_kv_j \in E(G)$ and $ j \geq i \}\cup \{v_{k}\}$.

\begin{lemma}
 \label{l:lem5}
Let G be a block graph, and $T(G)$ be a block tree of $G$. Let $\beta=(v_{1},v_{2},\ldots,v_{n})$ be the reverse of $BFS$-ordering of vertices in $T(G)$, then $v_iv_j \in E(G)$ implies $N_i[v_i] \subseteq N_i[v_j]$ where $j \geq i$.
\end{lemma}
\noindent
\begin{proof}
If $v_j$ is parent of $v_i$ in $T(G)$, then  a neighbor $v_k$ of $v_i$ with $k>i$ is a sibling of $v_i$ in $T(G)$. Hence, $v_kv_j \in E(G)$ and the result follows. If $v_j$ is a sibling of $v_i$ and $v_iv_j \in E(G)$, then $v_i$ and $v_j$ are in the same block $B$ of $G$. Now any neighbor $v_k$ of $v_i$ with $k>i$ is either a sibling or parent of $v_i$ in $T(G)$. If $v_k$ is parent of $v_i$ then $v_jv_k \in E(G)$ (as $v_j$ is a sibling of $v_i$). Next we consider that  $v_k$ is a sibling of $v_i$. Note that either $i<j<k$ or $i<k<j$, and the vertices $v_i, v_j,v_k$ appears in the reverse of BFS-ordering, which is also a BEO (by Lemma~\ref{l:lem4}). Since $v_iv_k, v_iv_j \in E(G)$, $v_jv_k \in E(G)$ by using the property of BEO.
\end{proof}

\begin{obs}
Let G be a block graph, and let $T(G)$ be a block tree of $G$. If $\beta=(v_{1},v_{2},\ldots,v_{n})$ is the reverse of $BFS$-ordering of vertices in $T(G)$, then for any descendant $v_r$ of $v_j$, $N_i(v_r) \subseteq N_i(v_j)$.
\end{obs}

\noindent Using the properties of a tree and Observation~\ref{ob:block1} we have the following important observation.

\begin{obs}If a vertex $v_i$ is at level $l+2$, then $v_i$ does not have any neighbor at level $l$. Also if $v_i$ is semipaired with $v_j$, then $v_j$ may be one of the following in $T(G)$.\\ $(i)$ Parent of $v_i$ that is $F(v_i)$ or a grand parent of $v_i$ that is $F(F(v_i))$.\\
 $(ii)$ Child of $v_i$ or a grand child of $v_i$. \\
  $(iii)$ Sibling of $v_i$ or sibling of $F(v_i)$.\\
   $(iv)$ Child of some sibling $v_s$ of $v_i$.
\end{obs}

Next, we present the detailed algorithm to compute a minimum cardinality semi-PD-set of a given block graph.

\begin{algorithm}[H]
\textbf{Input:} A block graph G=(V,E), corresponding block-tree $T(G)$, reverse of BFS ordering of vertices of $T(G)$: $\beta=(v_{1},v_{2},\ldots,v_{n})$\\
\textbf{Output:} A minimum cardinality semi-PD-set $D_{sp}$ of $G$.\\
\For{$i=1$ to $n$}{
	\If{$(D(v_i)=0$ and $i\neq n)$}{
		$D_{sp}=D_{sp}\cup\{F(v_i)\}$, let $F(v_i)=v_j$;\\
		$L(v_j)=1$;\\
		$D(u)=1 $ $\forall$ $u \in N_{G}[v_j]$;\\
		Let $C=\{u \in N_{G}[v_j] \mid m(u) \neq 0\}$;\\
		\If {$(C= \emptyset)$}{
			$m(F(v_j)) = j$;
			}
		\Else{
			Let $k=$min$\{b \mid v_b \in C\}$ and let $m(v_k)=r$;\\
			$L(v_j)= L(v_{r})=2$; $//$ semipair $v_{j}$ with $v_{r}$\\
			$m(v_k)=0$;
			}
		}
	\If{$(D(v_i) \neq 0$ and $m(v_i)=k \neq 0)$}{
		\If{$(L(F(v_i))=0)$}{
			$D_{sp}=D_{sp}\cup\{F(v_i)\}$;\\
			$D(u)=1 $ $\forall$ $u \in N_{G}[F(v_{i})]$;\\
			$L(v_k)= L(F(v_i))=2$ and $m(v_i)=0$;  $//$ semipair $v_{k}$ with $F(v_{i})$;
			}
		\ElseIf{$(L(v_i)=0)$}{
			$D_{sp}=D_{sp}\cup\{v_i\}$;\\
			$L(v_k)=L(v_i)=2$;  $//$ semipair $v_{k}$ with $v_{i}$\\
			$D(u)=1 $ $\forall$ $u \in N_{G}[v_{i}]$;\\
			$m(v_i)=0$;
			}
		\Else{
			$D_{sp}=D_{sp}\cup\{u\}$ where $u \in N(v_k)$ and $L(u) = 0$;\\
			$L(v_k)=L(u)=2$ and $m(v_i)=0$;  $//$ semipair $v_{k}$ with $u$;
			}
		}
	\If {$(D(v_n)=0)$}{
		$L(v_n)=2$;\\
		$D(v_n)=1$;\\
		$L(u)=2$ for some $u \in N_{G}(v_n)$ with $L(u)=0$;  $//$ semipair $v_{n}$ with $u$\\
        $D_{sp}=D_{sp}\cup \{v_{n},u\}$;\\
		}
}
return $D_{sp}$;

\caption{\textbf{Minimum Semipaired Domination in Block Graphs}}
\end{algorithm}

\subsection*{Illustration of the Algorithm with an example}
We illustrate the algorithm for computing a minimum cardinality semi-PD-set of the block graph shown in Fig.~\ref{fig:mix}. Since there are $20$ vertices in the graph, the algorithm will terminate in $20$-iterations. We process the vertices as they appear in the reverse of BFS-ordering. For the graph in the Fig.~\ref{fig:mix}, a reverse of BFS-ordering of the vertices is given by
\noindent $\beta = (u_{7}, u_{6}, u_{5}, u_{11}, u_{10}, u_{9}, u_{8}, u_{14},$ $u_{12}, u_{13}, u_{16}, u_{15},$ $u_{3}, u_{2}, u_{1}, u_{19}, u_{18}, u_{17}, u_{4}, u_{20})$.
The iterations of the algorithm are as follows:

\begin{center}
\begin{tabular}{ |c| }
\hline
\underline{\textsc{Initially}}\\$D(u_i) = L(u_i) = m(u_i)= 0$ for $i \in [20]$ and $D_{sp}= \emptyset $.\\
\hline
 \underline{\textsc{Iteration $1$}}\\
Since $D(u_7)=0$, we select $F(u_7)= u_{10}$.\\
$L(u_{10})=1$ and $D(u_{7})= D(u_{8})= D(u_{9})= D(u_{10})= D(u_{11})= D(u_{12})= 1$\\
As $C= \emptyset$, we have $m(F(u_{10}))=m(u_{12}) =10$\\
\textsc{After Iteration $1$}: $D_{sp}=\{u_{10}\}$ \\
\hline
 \underline{\textsc{Iteration $2$}}\\
Since $D(u_6)=0$, we select $F(u_6)= u_{8}$.\\
$L(u_{8})=1$ and $D(u_{5})= D(u_{6})= 1$\\
In this iteration note that $C = \{u_{12}\}$, $k=12$ and $m(u_{12})=10$.\\
$L(u_{10}) = L(u_{8}) = 2 $ and $m(u_{12}) = 0$\\
\textsc{After Iteration $2$}: $D_{sp}=\{u_{8}, u_{10}\}$ \\
\hline
We do not have any update in \underline{\textsc{Iteration $3, 4, 5, 6$, and $7$}}\\
\hline
 \underline{\textsc{Iteration $8$}}\\
Since $D(u_{14})=0$, we select $F(u_{14})= u_{15}$.\\
$L(u_{15})=1$ and $D(u_{14})= D(u_{15})= D(u_{16})= D(u_{17})= 1$\\
As $C= \emptyset$, hence $m(F(u_{15}))=m(u_{17}) = 15$\\
\textsc{After Iteration $8$}: $D_{sp}=\{u_{8}, u_{10}, u_{15}\}$\\
\hline
We do not have any update in \underline{\textsc{Iteration $9$}}\\
\hline
\underline{\textsc{Iteration $10$}}\\
Since $D(u_{13})=0$, we select $F(u_{13})= u_{19}$.\\
$L(u_{19})=1$ and $D(u_{13})= D(u_{17})= D(u_{18})= D(u_{19})= D(u_{20}) = 1$\\
As $C= \{u_{17}\}$, $k=17$ and $m(u_{17}) = 15$\\
$L(u_{15})=L(u_{19})=2$ and $m(u_{17})=0$\\
\textsc{After Iteration $10$}: $D_{sp}=\{u_{8}, u_{10}, u_{15}, u_{19}\}$\\
\hline
We do not have any update in \underline{\textsc{Iteration $11$ and $12$}}\\
\hline
\underline{\textsc{Iteration $13$}}\\
Since $D(u_{3})=0$, we select $F(u_{3})= u_{4}$.\\
$L(u_{4})=1$ and $D(u_{1})= D(u_{2})= D(u_{3})= D(u_{4})= 1$\\
As $C = \emptyset$, hence $m(F(u_{4})) = m(u_{20})=4$\\
\textsc{After Iteration $13$}: $D_{sp}=\{u_{8}, u_{10}, u_{15}, u_{19}, u_{4}\}$\\
\hline
We do not have any update in \underline{\textsc{Iteration $14$, $15$, $16$, $17$, $18$, and $19$}}\\
\hline
\underline{\textsc{Iteration $20$}}\\
In this iteration note that $m(u_{20}) = 4 \neq 0$.\\
As $L(F(u_{20}))= L(u_{20}) = 0$ hence, $L(u_{20}) = L(u_{4}) = 2 $ and $m(u_{20}) = 0$.\\
\textsc{After Iteration $20$}: $D_{sp}=\{u_{8}, u_{10}, u_{15}, u_{19}, u_{4}, u_{20}\}$\\
\bf{Return} $D_{sp}$\\
\hline
\end{tabular}
\\
\end{center}

\begin{lemma}
\label{l:lemma5}
For $1 \leq i \leq n$, when $v_i$ is the currently  considered vertex, the following holds:
\begin{enumerate}
\item $D(v_j)=1$ for $ j \in [i-1]$.
\item $m(v_j)=0$ for $ j \in [i-1]$.
\item if $v_i$ is at level $l$ from the root in $T(G)$, then for all vertices $u$ at level $l+2$ or more, either $L(u)=0$ or $2$.
\end{enumerate}
\end{lemma}

\begin{proof}
Note that in the $j^{th}$-iteration of the algorithm we dominate $v_j$ if it was not dominated before $j^{th}$-iteration, that is, if $D(v_j)=0$  then we include $F(v_j)$ in $D_{sp}$ to dominate $v_j$ and we also update $D(v_j)=1$. Also, in the $j^{th}$-iteration if $m(v_j)=r \neq 0$ then $v_r$ is already selected in $D_{sp}$ and we need to semipair $v_r$ with a vertex in $N_{T(G)}[v_j]$ or with some sibling of $v_j$. In the $j^{th}$-iteration we semipair $v_r$ with a vertex in $N_{T(G)}[v_j]$ or with some sibling of $v_j$ and update $m(v_j)=0$. Hence, the first two statements of the lemma follows.

Now, suppose if $v_i$ is at level $l$ from the root in $T(G)$ and there is a vertex $v_t$ at level $l+2$ or more such that $L(v_t)=1$. Hence $v_t$ is included in $D_{sp}$ but it is not semipaired yet. In this case, there must be a vertex $v_{t'}$ at level $l+1$ or more such that $m(v_{t'})=t$ and $t'<i$, which is a contradiction to the second statement of the lemma. Hence the result follows.
\end{proof}

Note that if for any vertex $v_r$, $m(v_r) = k \neq 0$ then $v_k$ is selected in the semi-PD-set  but still need to be semipaired with a vertex in $N_{G}[v_r] \setminus \{v_k\}$. Now for any vertex $v_t \in N_{G}[v_r] \setminus \{v_k\}$, if we update $L(v_t)=1$, that is if we select $v_t$ in the semi-PD-set, according to algorithm, we semipair $v_t$ with $v_k$ and make $m(v_r) = 0$. Hence, in any iteration if for any vertex $v_r$, $m(v_r) = k \neq 0$ we have the following observation.

\begin{obs}
\label{ob:block3}
In the $i^{th}$-iteration, if for any vertex $v_r$, $m(v_r) = k \neq 0$, then for any $u \in N_G[v_r] \setminus \{v_k\}$, $L(u)=0$ or $2$.
\end{obs}

For each $0 \leq i \leq n$, let $D_i=\{v \mid L(v) >0\}$ when $v_i$ has just been considered. In particular $D_0= \emptyset $ and $D_n = \{v \mid L(v) >0\}$, when all the vertices of the graph have been processed. Clearly, $D_n$ is a dominating set. Note that when we are processing the vertex $v_n$ in our algorithm, by Lemma~\ref{l:lemma5} for all vertices $u$ at level $2$ or more, either $L(u)=0$ or $2$ and $D(v_j)=1$ for $j \in [n-1]$. Also note that in the $n^{th}$-iteration, if for a vertex $u$ at level $1$, we have $L(u)=1$, then $m(v_n) \neq 0$ and hence by Observation~\ref{ob:block3} for all vertices $w$ other than $u$ at level $1$, we have $L(w)=0$ or $2$. In this iteration, if $L(v_n) = 0$, then we semipair $u$ with $v_n$; otherwise, we semipair $u$ with one of its neighbors. Also, if $D(v_n)=0$, this implies no vertex from level $1$ is selected in the dominating set. So, according to Algorithm~$1$, we select $v_n$ in the semi-PD-set and pair it with one of it's children. Therefore, after the $n^{th}$-iteration for all $u \in V(G)$, we have $L(u)=0$ or $2$ and $D_n$ is a semi-PD-set. So in order to prove the correctness of the algorithm, we need to show that $D_{n}$ is contained in a minimum semi-PD-set $D^*_{sp}$.

\begin{lemma}
\label{l:lem7}
For each $0 \leq i \leq n$, there is minimum semi-PD-set $D_{sp}'$ such that:
\begin{enumerate}
\item $D_i \subseteq D_{sp}'$ and if $u$ and $v$ are semipaired in $D_i$ then $u$ and $v$ are also semipaired in $D_{sp}'$.
\item if in the $i^{th}$-iteration we are updating $m(v_j)=k$ for some vertex $v_j$, then $v_k$ is either semipaired with a vertex in $N_{T(G)}[v_j]$ or with a sibling of $v_j$ in $D_{sp}'$.
\end{enumerate}
\end{lemma}

\begin{proof}
We will prove the result using induction on $i$. Clearly, when $i=0$, there is semi-PD-set  $D_{sp}'$ such that $D_0 \subseteq D_{sp}'$. Now, suppose that the result holds for any integer less than $i<n$, that is, there is a minimum semi-PD-set  $D_{sp}^*$ such that $D_{i-1} \subseteq D_{sp}^*$ and the second condition of the lemma is satisfied.\\

\noindent Now in the $i^{th}$-iteration we will have the following cases:\\

\noindent\textbf{Case 1.} $D(v_i) = 0$ and $i \neq n$.

In this case, $D_i=D_{i-1} \cup \{F(v_i)\}$. Let $F(v_i)=v_j$ and assume that $v_i$ is at level $l \geq 1$ from the root $v_n$. We proceed further with the following claim.

\begin{claim}
\label{c:cl1}
 Let $v_j \notin D_{sp}^*$, $v_p \in D_{sp}^*$ be a vertex dominating $v_i$, and $v_q$ be the vertex which is semipaired with $v_p$ in $D_{sp}^*$, then either $(D^*_{sp} \setminus \{v_q\}) \cup \{v_j\}$ or $(D^*_{sp} \setminus \{v_p\}) \cup \{v_j\}$ is a minimum semi-PD-set satisfying the first statement of the lemma.
\end{claim}

\begin{proof}
Suppose that $v_j \notin D_{sp}^*$. Let $v_p$ be a vertex in $D_{sp}^*$ dominating $v_i$, and let $v_q$ be the vertex that is semipaired with $v_p$. Clearly, $j \notin \{p,q\}$ and $v_p \notin D_{i-1}$. If $v_q \in D_{i-1}$ and $v_q$ is at level $l+2$ or more, then $v_p \in D_{i-1}$ by Lemma~\ref{l:lemma5}, a contradiction. Hence, if $v_q \in D_{i-1}$, then $v_q$ is at level $l+1$ or less from the root in $T(G)$. Further if $v_q \notin D_{i-1}$ and $v_q$ is at level $l+2$ or more in $T(G)$, then by Lemma~\ref{l:lemma5}, the set $D_{sp}'= (D^*_{sp} \setminus \{v_q\}) \cup \{v_j\}$ where $v_j$ is semipaired with $v_p$ is the required minimum semi-PD-set. Hence, we may assume that $v_q$ is at level $l+1$ or less.

Suppose firstly that the vertex $v_p$ is a child of $v_i$ in $T(G)$. Since $v_q$ is at level $l+1$ or less  and $q \neq j$, we have $v_q \in N_{G}[v_i] $. Hence, $d(v_j,v_q) \leq 2$. Note that $N_{i}(v_p)\subseteq N_{i}(v_j)$. Since $d(v_q,v_j)\leq 2$, the set $D_{sp}' = (D^*_{sp} \setminus \{v_p\}) \cup \{v_j\}$ where $v_j$ is semipaired with $v_q$ is the required minimum semi-PD-set.

Suppose secondly that $v_p$ is a sibling of $v_i$ or $v_p = v_i$ in $T(G)$. In this case, $v_p$ is a child of $v_j$ in $T(G)$. Since $v_q$ is at level $l+1$ or less hence, either $v_q$ is a child of $v_p$ or a child of some sibling say $v_s$ of $v_p$ or $v_q$ is a sibling of $v_p$ or a sibling of $v_j$ or $v_q=F(v_j)$.

If $v_q$ is a sibling of $v_p$ or a sibling of $v_j$ or $v_q=F(v_j)$, we have $v_qv_j \in E(G)$. If $v_q$ is a child of $v_p$, then $d(v_j,v_q) =2$. If $v_q$ is a child of some sibling, say $v_s$,  of $v_p$, then $v_s$ is a child of $v_j$. Thus since $v_q$ is a child of $v_s$ in $T(G)$, we again have $d(v_j,v_q) =2$. In all of the above cases, we have $d(v_j,v_q) \leq 2$. Further, since $v_p$ is a child of $v_j$, using Lemma~\ref{l:lemma5} and the fact that $N_{i}[v_p]\subseteq N_{i}[v_j]$, we have that $D_{sp}'= (D^*_{sp} \setminus \{v_p\}) \cup \{v_j\}$ where $v_j$ is semipaired with $v_q$ is the required minimum semi-PD-set.
\end{proof}

By Claim~\ref{c:cl1}, we may assume that $v_j \in D_{sp}^*$, for otherwise the desired result holds. We now let $C=\{u \in N_{G}[v_j] \mid m(u) \neq 0\}$,
and consider two subcases.\\

\noindent\textbf{Case 1.1.} $C = \emptyset$.

Here we are updating $m(F(v_j))=j$, we need to show that $v_j$ is either semipaired with a vertex in $N_{T(G)}[F(v_j)]$ or with some sibling of $F(v_j)$ in $D_{sp}^*$. Let $v_r$ be the vertex semipaired with  $v_j$ in $D_{sp}^*$ and suppose neither $v_r \notin N_{T(G)}[F(v_j)]$ nor $v_r$ is a sibling of $F(v_j)$ in $T(G)$. We note that in this case $v_r$ is at level $l$ or $l+1$ in $T(G)$. Further, since $C = \emptyset$, $m(u) = 0$ for all $u \in N_{G}[v_j]$. This implies that for all $u \in N[N[v_j]]$ (that is, for all $u$ such that $d(v_j,u) \leq 2$), $L(u) = 0$ or $2$. Thus, $v_r \in D_{i-1}$ implies $v_j \in D_{i-1}$ a contradiction. Hence, $v_r \notin D_{i-1}$. If $v_r$ is at level $l+1$ or a child of $v_j$, then using Lemma~\ref{l:lemma5} and the fact that $N_i(v_r) \subseteq N_i(v_j)$ we may conclude that $D_{sp}' = (D_{sp}^* \setminus \{v_r\}) \cup \{F(v_j)\}$ with $v_j$ semipaired with $F(v_j)$ in $D_{sp}'$ is the required minimum semi-PD-set. If $v_r$ is at level $l$  and not a child of $v_j$ in $T(G)$, then $v_r$ is a child of some sibling $v_s$ of $v_j$ in $T(G)$ such that $v_sv_j \in E(G)$. In this case $N_i[v_r] \subseteq N_i[v_s]$. Now if $v_s \in D_{sp}^*$ then using Lemma~\ref{l:lemma5}, $D_{sp}' = (D_{sp}^* \setminus \{v_r\}) \cup \{F(v_j)\}$ where $v_j$ is semipaired with $F(v_j)$ in $D_{sp}'$ is the required minimum semi-PD-set. If $v_s \notin D_{sp}^*$, then using Lemma~\ref{l:lemma5}, $D_{sp}' = (D_{sp}^* \setminus \{v_r\}) \cup \{v_s\}$ where $v_j$ is semipaired with $v_s$ in $D_{sp}'$ is the required minimum semi-PD-set.\\

\noindent\textbf{Case 1.2.} $C \ne \emptyset$.

Let $k = \min \{b \mid v_b \in C\}$ and let $m(v_k)=r$. Since in this case we semipair vertices $v_r$ and $v_j$ in $D_i$, we need to show that they are semipaired in $D_{sp}'$ as well. By Lemma~\ref{l:lemma5}, $k \geq i$. If $v_r$ is semipaired with $v_j$ in $D^*_{sp}$ then the result follows; otherwise, let $v_s$ and $v_t$ be the vertices semipaired with $v_j$ and $v_r$, respectively, in $D_{sp}^*$.
Since $m(v_k) = r$ after $(i-1)^{th}$-iteration, $v_r$ still needs to be semipaired with a vertex in $N_{T(G)}[v_k]$ or with some sibling of $v_k$. Hence, $v_r$ is not semipaired until the $(i-1)^{th}$-iteration in $D_{i-1}$. Also as $m(v_k) = r$, using the second part of Lemma~\ref{l:lem7}, either $v_t \in N_{T(G)}[v_k]$ or $v_t$ is a sibling of $v_k$ such that $v_tv_k \in E(G)$, that is, $v_t \in N_G[v_k]$. Note that if $v_t \in D_{i-1}$ then $L(v_t) =1$ but as $m(v_k) = r$ by Observation~\ref{ob:block3}, for all $u \in N_G[v_k] \setminus \{v_r\}$, $L(u)=0$ or $2$. Hence, $v_t \notin D_{i-1}$.

Suppose that $v_k$ is a child of $v_j$ in $T(G)$. In this case, $v_t$ is a child of either $v_j$ or $v_k$ in $T(G)$. If $v_t$ is a child of $v_k$ then $N_i(v_t) \subseteq N_i[v_j]$ and if $v_t$ is a child of $v_j$ then $N_i[v_t] \subseteq N_i[v_j]$. Now, if $N_{G}(v_s) \subseteq D^*_{sp}$, then by Lemma~\ref{l:lemma5} $D'_{sp} = D^*_{sp} \setminus \{v_s, v_t\}$ is semi-PD-set of smaller size, a contradiction. Hence there exist a vertex $u \in N_{G}(v_s)$ such that $u \notin D^*_{sp}$ and $D'_{sp} = (D^*_{sp} \setminus \{v_t\}) \cup \{u\}$ with $v_j$ and $v_s$ semipaired with $v_r$ and $u$, respectively, is the required minimum semi-PD-set. Hence, we may assume that $v_k$ is not a child of $v_j$ in $T(G)$, for otherwise the desired result follows. Also, since $v_k$ is the minimum index vertex such that $v_k \in C$, in further cases we may assume that there is no child $v_{k'}$ of $v_j$ with $m(v_{k'}) \neq 0$. So, $v_k$ can be one of the following: $(a)$ a sibling of $v_j$, $(b)$ $v_j$, or $(c)$ $F(v_j)$. We will consider these remaining cases under inclusion or exclusion of $v_s$ in $D_{i-1}$.\\

\noindent\textbf{Case 1.2.1}  $v_s \in D_{i-1}$.

If $L(v_s) = 2$, then $v_j \in D_{i-1}$, a contradiction. Hence, $L(v_s)=1$. Thus, if $v_s \in D_{i-1}$, then there exists a vertex $v_{s_1}$ such that $m(v_{s_1})=s$, where $v_{s_1} = F(v_s)$ and $v_{s_1} \in N_{G}[v_j]$. If $v_s$ is at level $l+1$ in $T(G)$, then $v_{s_1}$ is a child of $v_j$, a contradiction noting that $k = \min \{s \mid v_s \in C\}$. If $v_s$ is at level $l$ in $T(G)$, then  $v_{s_1}$ is a sibling of $v_j$ in $T(G)$ and $s_1 > k$, implying that $v_k \neq F(v_j)$. If $v_s$ is at level $l-1$, then $v_{s_1}=F(v_j)$ hence, $v_k \neq F(v_j)$. Therefore, we may note that if $v_s \in D_{i-1}$ then $v_k \neq F(v_j)$. So $v_k$ may be a sibling of $v_j$ or $v_k=v_j$.

Let $v_k$ is a sibling of $v_j$. Now, $v_kv_j \in E(G)$ and either $s_1=j$ or $v_{s_1}v_j \in E(G)$. If $v_t$ is a child of $v_k$, then using the fact that $N_i[v_t] \subseteq N_i[v_k]$, the set $D_{sp}' = (D_{sp}^* \setminus \{v_t\}) \cup \{v_k\}$ with $v_j$ semipaired with $v_r$ and with $v_s$ semipaired with $v_k$ is the required minimum semi-PD-set. If $v_t$ is not a child of $v_k$, then $v_t$ is either a sibling of $v_k$ such that $v_kv_t \in E(G)$ or $v_t = F(v_j)$.  We note that if either $v_t$ is a sibling of $v_k$ such that $v_kv_t \in E(G)$ or $v_t = F(v_j)$ then the vertices in the set $\{v_j,v_k,v_{s_1},v_t\}$ are in the same block of $G$ and $v_sv_{s_1} \in E(G)$. Hence, $d(v_s,v_t) \leq 2$.
Therefore, if we exchange the semipairs, that is, if we semipair $v_j$ and $v_r$, and semipair $v_s$ and $v_t$ in the $D_{sp}^*$, then $D_{sp}^*$ is the desired minimum semi-PD-set.

Let $v_j=v_k$. If $v_t$ is a child of $v_k$ and $N_{G}(v_s) \subseteq D_{sp}^*$ then using Lemma~\ref{l:lemma5} and the fact that $N_i[v_t] \subseteq N_i[v_k]$, the set $D'_{sp} = D^*_{sp} \setminus \{v_s, v_t\}$ is a semi-PD-set of smaller size, a contradiction. Therefore, if $v_t$ is a child of $v_k$, then there exists a vertex $u \in N_{G}(v_s)$ such that $u \notin D^*_{sp}$. In this case, using Lemma~\ref{l:lemma5} and the fact that $N_i[v_t] \subseteq N_i[v_k]$, the set $D_{sp}' = (D_{sp}^* \setminus \{v_t\}) \cup \{u\}$ with $v_j$ semipaired with $v_r$ and with $v_s$ semipaired with $u$ is the required minimum semi-PD-set. If $v_t$ is not a child of $v_k$ in $T(G)$, then observe that $d(v_s,v_t) \leq 2$. Hence, if we exchange the semipairs, that is, we semipair $v_j$ and $v_r$, and semipair $v_s$ and $v_t$ in the $D_{sp}^*$, then $D_{sp}^*$ is the desired minimum semi-PD-set.

\noindent\textbf{Case 1.2.2.}  $v_s \notin D_{i-1}$.

If $s \leq i$ or $v_s$ is a child of $v_j$ in $T(G)$ and $N_{G}(v_t) \subseteq D^*_{sp}$, then  $D'_{sp} = D^*_{sp} \setminus \{v_s, v_t\}$ is a semi-PD-set of smaller size, a contradiction noting that $N_i(v_s) \subseteq N_i(v_j)$. Hence, if $s \leq i$ or $v_s$ is a child of $v_j$ in $T(G)$, then there exist a vertex $u \in N_{G}(v_t)$ such that $u \notin D^*_{sp}$ and in this case $D_{sp}' = (D^*_{sp} \setminus \{v_s\}) \cup \{u\}$ with $v_j$ and $v_t$ semipaired with $v_r$ and $u$, respectively, is the required minimum semi-PD-set. Now we have, $s > i$ and $v_s$ is not a child of $v_j$, implying that in $T(G)$, the vertex $v_s$ may be one of the following: (i) a child of some sibling $v_s'$ of $v_j$ in $T(G)$, (ii) $v_s$ is a sibling of $v_j$, (iii) $v_s=F(v_j)$, (iv) $v_s$ is a sibling of $F(v_j)$ or (v) $v_s = F(F(v_j))$.

Suppose firstly that $v_k$ is a sibling of $v_j$ in $T(G)$. If $v_t$ is a child of $v_k$, then using the fact that $N_i[v_t] \subseteq N_i[v_k]$, the set $D_{sp}' = (D_{sp}^* \setminus \{v_t\}) \cup \{v_k\}$ with $v_j$ semipaired with $v_r$, and $v_s$ semipaired with $v_k$ is the required minimum semi-PD-set. If $v_t$ is not a child of $v_k$, then $v_t$ is either a sibling of $v_k$ or $v_t = F(v_j)$. If $v_t$ is a sibling of $v_k$ or $v_t = F(v_j)$ then vertices $v_j$, $v_k$, and $v_t$ belongs to the same block of $G$ and in both cases, we observe that $d(v_s,v_t) \leq 2$. Hence, if we exchange the semipairs, that is, if we semipair $v_j$ and $v_r$, and semipair $v_s$ and $v_t$ in the $D_{sp}^*$, then $D_{sp}^*$ is the desired minimum semi-PD-set.

Suppose secondly that $v_j=v_k$. If $v_t$ is a child of $v_k$ and $N[v_s] \subseteq D_{sp}^*$, then using the fact that $N_i[v_t] \subseteq N_i[v_k]$, the set $D'_{sp} = D^*_{sp} \setminus \{v_s, v_t\}$ is semi-PD-set of smaller size, a contradiction. Therefore, if $v_t$ is a child of $v_k$, then there exists a vertex $u \in N(v_s)$ such that $u \notin D^*_{sp}$. In this case, using Lemma~\ref{l:lemma5} and the fact that $N_i[v_t] \subseteq N_i[v_k]$, the set $D_{sp}' = (D_{sp}^* \setminus \{v_t\}) \cup \{u\}$ with $v_j$ semipaired with $v_r$, and with $v_s$ semipaired with $u$, is the required minimum semi-PD-set. If $v_t$ is not a child of $v_k$ in $T(G)$, then observe that $d(v_s,v_t) \leq 2$. Hence, if we exchange the semipairs, that is, if we semipair $v_j$ and $v_r$, and semipair $v_s$ and $v_t$ in the $D_{sp}^*$, then $D_{sp}^*$ is the desired minimum semi-PD-set.

Suppose next that $v_k=F(v_j)$. If $v_s$ is a child of some sibling $v_{s_1}$ of $v_j$ in $T(G)$ then $v_{s_1}$ is a child of $F(v_j)$. We note that as $v_k=F(v_j)$, this implies $v_t \in N_{T(G)}(F(v_j))$ or $v_t$ is a sibling of $F(v_j)$ and $v_tF(v_j) \in E(G)$. Hence we may observe that $d(v_{s_1},v_t) \leq 2$. Therefore, if $v_s$ is a child of some sibling $v_{s_1}$ of $v_j$ in $T(G)$, then using Lemma~\ref{l:lemma5} and the fact that $N_i[v_s] \subseteq N_i[v_{s_1}]$, the set $D_{sp}' = (D_{sp}^* \setminus \{v_s\}) \cup \{v_{s_1}\}$ with $v_j$ semipaired with $v_r$, and with $v_t$  semipaired with $v_{s_1}$, is the required minimum semi-PD-set. If $v_s$ is not a child of some sibling $v_{s_1}$ of $v_j$ in $T(G)$, then we observe that $d(v_s,v_t) \leq 2$. Hence, if we exchange the semipairs, that is, if we semipair $v_j$ and $v_r$, and semipair $v_s$ and $v_t$ in the $D_{sp}^*$, then $D_{sp}^*$ is the desired minimum semi-PD-set.\\

\noindent\textbf{Case 2.} $D(v_i) \neq 0$ and $m(v_i)=k \neq 0$.

In this case, we semipair $v_k$, either with some vertex in $N_{T(G)}[v_i]$ or with a sibling of $v_i$. Suppose $w$ is the vertex that is semipaired with $v_k$ in $D_{i-1}$. So we need to show that $v_k$ is semipaired with $w$ in $D_{sp}'$ as well. Using induction, there is minimum semi-PD-set $D_{sp}^*$, such that $D_{i-1} \subseteq D_{sp}^*$ and $v_k$ is semipaired with a vertex either in $N_{T(G)}[v_i]$ or a sibling of $v_i$. Using Observation~\ref{ob:block3}, we note that for all vertices $u \in N_G[v_i] \setminus \{v_k\}$ either $L(u)=0$ or $2$ and $D(v_j)=1$ for $j \in [i-1]$ by  Lemma~\ref{l:lemma5}.  Let $v_k$ be semipaired with $v_p$. We note that $v_p \notin D_{i-1}$.\\

\noindent\textbf{Case 2.1.} $L(F(v_i))=0$.

If $v_p=F(v_i)$, then the result follows. Let $v_p \neq F(v_i)$, implying that $v_p$ is a child of $v_i$ or a sibling of $v_i$ or $p=i$. If $v_p$ is child of $v_i$ then $N_{i}(v_p) \subseteq N_{i}[F(v_i)]$  and if $v_p$ is a sibling of $v_i$ or $p=i$ then we note that $N_{i}[v_p] \subseteq N_{i}[F(v_i)]$. Using Lemma~\ref{l:lemma5}, we can update  $D_{sp}'=(D^*_{sp} \setminus \{v_p\}) \cup \{F(v_i)\}$ with $v_k$ semipaired with $F(v_i)$ to get the required minimum semi-PD-set.\\

\noindent\textbf{Case 2.2.} $L(F(v_i)) \neq 0$ and $L(v_i)=0$.

In this case, $F(v_i) \in D_{i-1}$ and $v_p \neq F(v_i)$. If $v_p=v_i$ in $D_{sp}^*$, then $D_{sp}^*$ is the desired set. Now suppose $v_p \neq v_i$, implying that $v_p$ is either a child or a sibling of $v_i$. Similar to the previous case, we note that if $v_p$ is child of $v_i$ then $N_{i}(v_p) \subseteq N_{i}[F(v_i)]$  and if $v_p$ is a sibling of $v_i$ or $p=i$ then $N_{i}[v_p] \subseteq N_{i}[F(v_i)]$. Using Lemma~\ref{l:lemma5} we can update  $D_{sp}' = (D^*_{sp} \setminus \{v_p\}) \cup \{v_i\}$ with $v_k$ semipaired with $v_i$ to get the required minimum semi-PD-set.\\

\noindent\textbf{Case 2.3.} $L(F(v_i)) = L(v_i) \neq 0$.

  In this case, $F(v_i), v_i \in D_{i-1}$ and neither $F(v_i) =  v_p$ nor $v_i =  v_p$. Here, $v_p$ is either a child or a sibling of $v_i$ and similar to the previous cases, we note that $N_{i}[v_p] \subseteq N_{i}[F(v_i)] \cup N_i[v_i]$. Since $v_p \notin D_{i-1}$, we have $L(v_p) =0$. Therefore, $D_{sp}^*$ with $v_k$ semipaired with $v_p$ is the desired minimum semi-PD-set.\\

\noindent\textbf{Case 3.} $D(v_n)=0$.

Since $D(v_n) = 0$,  $v_n$ is not dominated by $D_{n-1}$.  Hence, no vertex from $N_{G}[v_n]$ is selected in $D_{n-1}$ and $L(u)=0$ for all $u \in N_{G}[v_n]$. Let $v_p$ be the vertex dominating $v_n$ in $D_{sp}^*$ and $v_q$ be the vertex semipaired with $v_p$ in $D_{sp}^*$. If $v_q\in D_{n-1}$, then $L(v_{q})=1$. Hence, by Lemma~\ref{l:lemma5}, $v_q$ is at level $1$ or less in $T(G)$, a contradiction as $D(v_n)=0$. If $v_q \notin D_{n-1}$, then using the fact that $L(u)=0$ for all $u \in N_{G}[v_n]$ and by Lemma~\ref{l:lemma5}, we can state that $D_{sp}'=(D_{sp}^* \setminus \{v_p,v_q\}) \cup \{v_n,u\}$ where $u \in N_G(v_n)$ is the desired minimum semi-PD-set.

This completed the proof of Lemma~\ref{l:lem7}.
\end{proof}

Now, we are ready to state the main result of this Section.

\begin{theorem}
Given a block graph $G$, a minimum semi-PD-set of $G$ can be computed in $O(n+m)$-time.
\end{theorem}
\begin{proof}
By Lemma~\ref{l:lem7}, we claim that there is a minimum semi-PD-set $D_{sp}^*$ such that $D_n \subseteq D_{sp}^*$ where $D_n$ is the semi-PD-set returned by the Algorithm~$1$. This proves that  the set $D_n$ returned by Algorithm~$1$ is a minimum semi-PD-set of $G$. Next, we analyze the complexity of computing $D_{n}$ for a given block graph $G$.

By Theorem~\ref{th:6}, given a block graph $G=(V,E)$, a BEO of vertices of $G$ can be computed in $O(n+m)$-time, and the corresponding block tree can also be constructed in $O(n+m)$-time. Now, given a block tree $T$, we can find the reverse of BFS ordering of $T$ in $O(n+m)$-time. Also, all the computations in Algorithm~$1$ can be performed in $O(n+m)$-time. This proves that a minimum semi-PD-set of any block graph can be computed in $O(n+m)$-time.
\end{proof}

\section{APX-completeness for Bounded Degree Graphs}
\label{Sec:5}
In this section, we show that the \textsc{Minimum Semipaired Dominaion} problem is APX-complete for graphs with maximum degree $3$. It is known that the \textsc{Minimum Semipaired Domination} problem for a graph $G$ with maximum degree $\Delta$ can be approximated with an approximation ratio of $1+\ln(2\Delta+2)$~\cite{iwoca}. Hence the \textsc{Minimum Semipaired Domination} problem is in APX for bounded degree graphs. To show APX-completeness, we use the concept of L-reduction. First, we recall the definition of L-reduction.

\begin{definition} Given two NP optimization problems $F$ and $G$ and a polynomial time transformation $f$ from instances of $F$ to instances of $G$, we say that $f$ is an L-reduction if there are positive constants $\alpha$ and $\beta$ such that for every instance $x$ of $F$ the following holds.
\begin{enumerate}
  \item $opt_{G}(f(x)) \le  \alpha \cdot opt_{F}(x)$.
  \item for every feasible solution $y$ of $f(x)$ with objective value $m_{G}(f(x),y)=c_{2}$
we can in polynomial time find a solution $y'$ of $x$ with
$m_{F}(x,y')=c_{1}$ such that $|opt_{F}(x)-c_{1}| \le \beta
|opt_{G}(f(x))-c_{2}|$.
\end{enumerate}

To show the APX-completeness of a problem $\Pi \in $ APX, it suffices to show that there is an L-reduction from some APX-complete problem to $\Pi$.
\end{definition}

We first show that the \textsc{Minimum Semipaired Dominaion} problem is APX-complete for graphs with maximum degree $4$. To show this result, we prove that the reduction given in the proof of Theorem $2$ in \cite{iwoca} is an L-reduction. So, we show an L-reduction from the \textsc{Minimum Vertex Cover problem} for  graphs with maximum degree $3$~\cite{alimonti}. The \textsc{Minimum Vertex Cover} problem is already known to be APX-complete for graphs with maximum degree $3$. For a graph $G=(V,E)$, a set $S\subseteq V$ is called a \emph{vertex cover} of $G$ if for every edge $e=uv\in E$, $S\cap\{u,v\}\neq \emptyset$. For a graph $G$, the \textsc{Minimum Vertex Cover} problem is to find a vertex cover of $G$ of minimum cardinality.
 Next, we present an L-reduction from the \textsc{Minimum Semipaired Domination} problem for graphs with maximum degree $4$ to the \textsc{Minimum Semipaired Domination} problem for graphs with maximum degree $3$.

\begin{theorem}
\label{l:apx4}
The \textsc{Minimum Semipaired Domination} problem is APX-complete for graphs with maximum degree $4$.
\end{theorem}

\begin{proof}
The  \textsc{Minimum Semipaired Domination} problem is in APX for graphs with maximum degree $4$. Hence, to prove the APX-completeness, it is sufficient to give an L-reduction $f$, from the set of instances for the \textsc{Minimum Vertex Cover problem} for  graphs with maximum degree $3$, to the set of instances for the \textsc{Minimum Semipaired Domination} problem for graphs with maximum degree $4$.

Given a graph $G=(V,E)$, where $V=\{v_{1},v_{2},\ldots,v_{n}\}$ with $d_{G}(v_{i})\leq 3$ for each  $i\in [n]$ and $E=\{e_{1},e_{2},\ldots,e_{m}\}$, we construct a graph $G'=(V',E')$ as follows:

Let $V_{k}=\{v_{i}^{k} \mid i \in [n]\}$, $E_{k}=\{e_{j}^{k}\mid j \in [m]\}$ for $k\in [2]$ and $X=\{w_{i}, x_{i}, y_{i}, z_{i}\mid i \in [n]\}$.
Define $V(G')= V_{1}\cup V_{2}\cup E_{1}\cup E_{2} \cup X$,
and $E(G')=\{v_{i}^{1}w_{i}, v_{i}^{2}w_{i}, w_{i}x_{i}, x_{i}y_{i}, y_{i}z_{i} \mid i \in [n]\} \cup \{v_{i}^{l}e_{j}^{l},v_{k}^{l}e_{j}^{l}\mid l\in [2], j \in [m]$ and $e_j = v_{i}v_{k}  \in E\}$. Fig.~\ref{fig:graph} illustrates the construction of $G'$ from $G$. Note that if the degree of a vertex in $G$ is bounded by $3$, then a vertex in $G'$ has degree at most $4$.

 \begin{figure}[htbp]
  \begin{center}
    \includegraphics[width=0.98\textwidth]{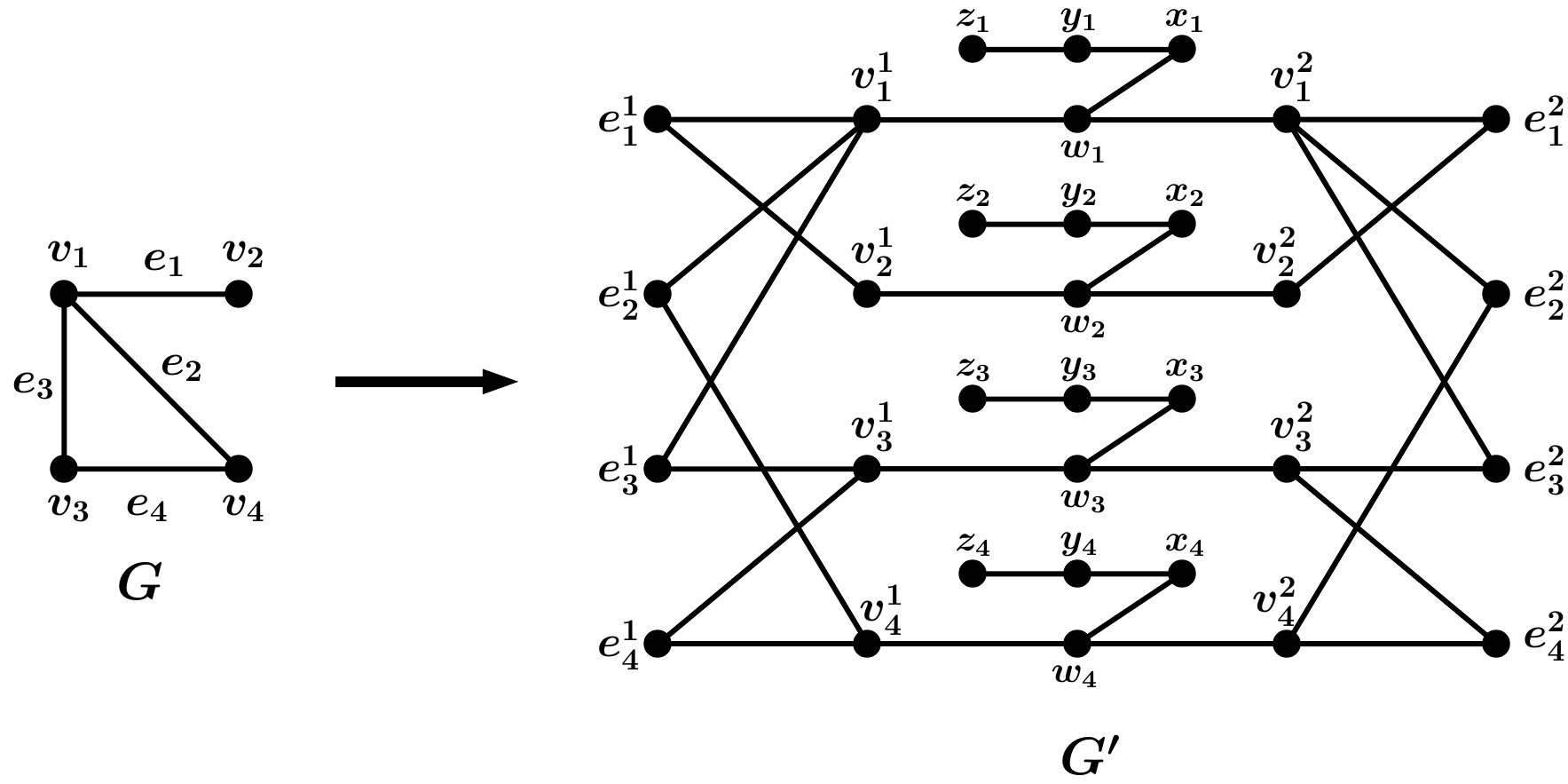}
    \caption{An illustration of the construction of $G'$ from $G$ in the proof of Theorem~\ref{l:apx4}.}
    \label{fig:graph}
  \end{center}
\end{figure}

Next, we show that the above reduction is an L-reduction. The following claim is enough to complete the proof of the theorem.

\begin{claim}\label{c:claim2}
If $V_{c}^{*}$ denotes a minimum vertex cover of $G$ and $D_{sp}^{*}$ denotes a minimum cardinality semi-PD-set of $G'$, and $n$ denotes the number of vertices of  $G$, then $|D_{sp}^{*}| = 2|V_{c}^{*}|+2n$. Further, if $D_{sp}$ is an arbitrary semi-PD-set of $G'$, then we can construct a vertex cover $V_{c}$ of $G$, such that $|V_{c}|-|V_{c}^{*}| \leq |D_{sp}|-|D_{sp}^{*}|$.
 \end{claim}
\begin{proof}
Let $V_{c}^{*}$ denotes a minimum vertex cover of $G$. Then, the set
\[
D_{sp}=\{v_{i}^{1},v_{i}^{2}\mid v_{i}\in V_{c}^{*} \} \cup \{w_i, y_i \mid i \in [n]\}
\]
is a semi-PD-set of $G'$ which implies that $|D_{sp}| \leq 2|V_{c}^{*}| +2n$. Hence, if $D_{sp}^{*}$ denotes a semi-PD-set of $G'$ of minimum cardinality, then $|D_{sp}^{*}|\leq 2|V_{c}^{*}|+2n$.

Next, suppose that $D_{sp}$ is an  arbitrary semi-PD-set of $G'$. Note that $|D_{sp} \cap \{w_{i},x_{i},y_{i},z_{i}\}|\geq 2$ for each $i\in [n]$. Hence, without loss of generality, we may assume that $\{w_{i},y_{i}\mid i \in [n]\} \subseteq D_{sp}$ (with $w_i$ and $y_i$ paired in $D_{sp}$). Let
\[
D = D_{sp} \setminus \{w_{i},y_{i}\mid i \in [n]\}.
\]

We note that $|D| = |D_{sp}| - 2n$. Let $D_1 = D\cap (V_{1}\cup E_{1})$ and $D_2 = D\cap (V_{2}\cup E_{2})$. Renaming the sets if necessary, we may assume that $|D_1| \leq |D_2|$, implying that $|D_1| \leq |D|/2$. In order to dominate a vertex $e_{i}^{1}\in E_{1}$, either $e_{i}^{1} \in D_{sp}$ or $v_{j}^{1}\in D_{sp}$ where $v_{j}^{1}\in N_{G'}(e_{i}^{1})$. If $N_{G'}(e_{i}^{1}) \cap D_{sp} = \emptyset$, then we update $D_{1}$ as $D_1 = (D_1 \setminus \{e_{i}^{1}\}) \cup \{v_{j}^{1}\}$ for some $v_{j}^{1}\in N_{G'}(e_{i}^{1})$. We note that the cardinality of the set $D_1$ remains unchanged after updating $D_1$ for all such $e_{i}^{1}$, and so, $|D_{1}| \leq |D|/2 = (|D_{sp}|-2n)/2$. Also every vertex in $E_1$ is now dominated by some vertex in $D_1$. Therefore the set
\[
V_{c}=\{v_{i}\mid v_{i}^{1}\in D_{1}\}
\]
is a vertex cover of $G$ and $|V_{c}|=|D_{1}| \leq \frac{1}{2}(|D_{sp}|-2n)$. Hence, $|D_{sp}| \geq 2 |V_{c}|+2n$.
Now, if $V_{c}^{*}$ is a minimum vertex cover of $G$, then $|D_{sp}| \ge 2 |V_{c}^{*}|+2n$. This is true for every semi-PD-set, $D_{sp}$, of $G'$. In particular, if $D_{sp}^{*}$ is a minimum semi-PD-set of $G'$, then we have $|D_{sp}^{*}| \geq 2|V_{c}^{*}|+2n$. As observed earlier, $|D_{sp}^{*}| \leq 2|V_{c}^{*}| +2n$. Consequently, $|D_{sp}^{*}| = 2|V_{c}^{*}|+2n$. Further,
\[
|V_{c}| - |V_{c}^{*}|
\leq  \frac{1}{2}(|D_{sp}|-2n) - \frac{1}{2}(|D_{sp}^{*}| - 2n)
= \frac{1}{2}(|D_{sp}|-|D_{sp}^{*}|).
\]
\end{proof}

Since the maximum degree of $G$ is $3$ and $G$ is connected, $n-1 \leq m \leq 3| V_{c}^{*}|$. Therefore, $|D_{sp}^{*}|=2|V_{c}^{*}|+2n \leq 2| V_{c}^{*}|+2(3| V_{c}^{*}|+1) \leq 8| V_{c}^{*}| +2 \leq 10| V_{c}^{*}|$. This proves that $f$ is an L-reduction with $\alpha =10$ and $\beta=\frac{1}{2}$.
\end{proof}

We observe that the graph $G'$ constructed in Theorem~\ref{l:apx4} is also a bipartite graph. Hence, as a corollary of Theorem~\ref{l:apx4}, we have the following result.

\begin{cor}
\label{corollary1}
The \textsc{Minimum Semipaired Domination} problem is APX-complete for bipartite graphs with maximum degree $4$.
\end{cor}

\begin{theorem}
\label{t:apx3}
The \textsc{Minimum Semipaired Domination} problem is APX-complete for graphs with maximum degree $3$.
\end{theorem}
\begin{proof}
To prove the theorem, it is sufficient to describe an L-reduction $h$, from the \textsc{Minimum Semipaired Domination} for graphs with maximum degree $4$ to the \textsc{Minimum Semipaired Domination} for graphs with maximum degree $3$.

Given a graph $G$ with maximum degree $4$, we construct a graph $G'$ with maximum degree $3$ as follows: for a vertex $v\in V(G)$ with $d_G(v) = 4$, we split and transform $v$ as illustrated in Fig.~\ref{fig:deg4}. For a vertex $v$ with $d_G(v) \leq 3$, we do not perform any transformation.

 \begin{figure}[htbp]
  \begin{center}
    \includegraphics[width=0.98\textwidth]{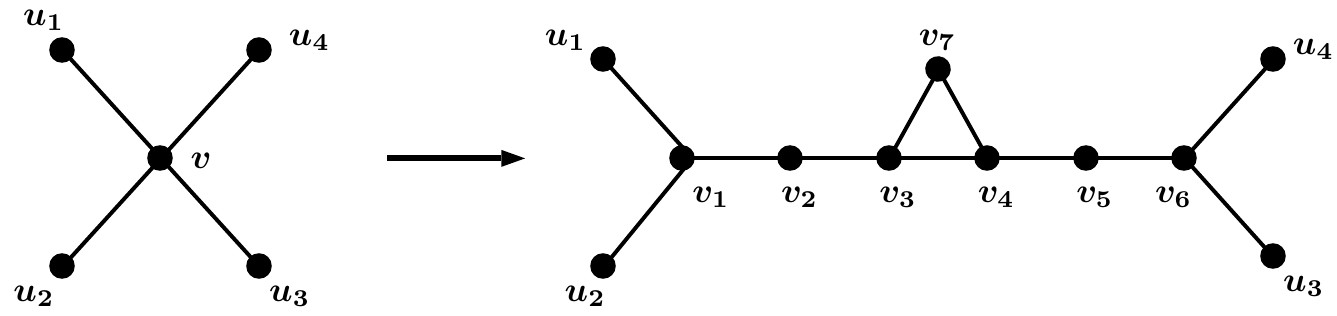}
    \caption{Transformation of a vertex $v \in V(G)$ with $d_{G}(v)$ = 4.}
    \label{fig:deg4}
  \end{center}
\end{figure}

By construction, for every vertex $v$ in $G'$, we have $d_{G'}(v) \leq 3$. Next, we show that the above reduction is an L-reduction. The following claim is enough to complete the proof of the theorem.

\begin{claim}\label{c:claim3}
If $D_{sp}^{*}$ and $D_{sp}'^{*}$ are minimum semi-PD-sets of $G$ and $G'$, respectively, and $k$ denotes the number of vertices of degree~$4$ in $G$, then $|D_{sp}'^{*}| = |D_{sp}^{*}|+2k$. Further, if $D_{sp}'$ is an arbitrary semi-PD-set of $G'$, then we can construct a semi-PD-set $D_{sp}$ of $G$, such that $|D_{sp}|-|D_{sp}^{*}| \leq |D_{sp}'|-|D_{sp}'^{*}|$.
 \end{claim}
\begin{proof}
Let $D_{sp}$ be a semi-PD-set of $G$. We construct a semi-PD-set $D_{sp}'$ of $G'$ as follows:
\begin{enumerate}
\item[(1)] If $d_G(v) \leq 3$, $v \in D_{sp}'$ if and only if $v \in D_{sp}$.
\item[(2)] If $d_{G}(v) = 4$, then we proceed as follows.
\begin{enumerate}
\item[(2.1)] If $v \in D_{sp}$ and $v$ is semipaired with a vertex $u \in N[w]$ such that $w \in \{u_1, u_2\}$, take $v_1, v_4, v_6$ in $D_{sp}'$.
\item[(2.2)] If $v \in D_{sp}$ and $v$ is semipaired with a vertex $u \in N[w]$ such that $w \in \{u_3, u_4\}$, take $v_1, v_3, v_6$ in $D_{sp}'$.
\item[(2.3)] If $v \notin D_{sp}$ and $v$ is dominated by a vertex in the set $\{u_1, u_2\}$, take $v_3, v_5$ in $D_{sp}'$.
\item[(2.4)] If $v \notin D_{sp}$ and $v$ is dominated by a vertex in the set $\{u_3, u_4\}$, take $v_2, v_4$ in $D_{sp}'$.
\end{enumerate}
\end{enumerate}
Let $k$ be the number of vertices of degree $4$ in $G$. We observe that, $D_{sp}'$ is a semi-PD-set of the transformed graph $G'$ and $|D_{sp}'| = |D_{sp}|+2k$. Thus, $|D_{sp}'^{*}| \leq |D_{sp}^{*}|+2k$, where $D_{sp}'^{*}$ and $D_{sp}^{*}$ are minimum semi-PD-sets of $G'$ and $G$, respectively.

Conversely, let $D_{sp}'$ be a semi-PD-set of $G'$. Now we construct a semi-PD-set $D_{sp}$ of $G$ from $D_{sp}'$. If $d_G(v) \leq 3$, then we will include $v$ in $D_{sp}$ if and only if $v \in D_{sp}'$. If $d_G(v) = 4$, then $v$ is transformed as shown in Fig.~\ref{fig:3}. For each vertex $v \in V(G)$ of degree~$4$, let $\phi(v)=|\{v_1,v_2, \ldots, v_7\} \cap D_{sp}'|$. If the degree of $v$ in $G$ is $4$, then we include $v$ in $D_{sp}$ if and only if $\phi(v) \geq 3$. Without loss of generality, we may assume that, to dominate $v_7$ either $v_3 \in D_{sp}'$ or $v_4 \in D_{sp}'$. Also, $v_3$ and $v_4$ can be semipaired only with the vertices in the set $\{v_1, v_2, \ldots , v_7\}$. Hence, $\phi(v) \geq 2$.

We note that the vertices in the set $\{v_2,v_3,v_4,v_5,v_7\}$ cannot be dominated by a vertex $w \in D_{sp}' \setminus \{v_1, v_2 , \ldots , v_7\}$. Hence, if $\phi(v) = 2$, then without loss of generality we may assume that either $\{v_3,v_5\} \subseteq D_{sp}'$ and $v_3$ is semipaired with $v_5$ in $D_{sp}'$ or  $\{v_2,v_4\} \subseteq D_{sp}'$ and $v_4$ is semipaired with $v_2$ in $D_{sp}'$. Note that in either case, $V(G') \setminus \{v_1,v_2, \ldots, v_7\}$ is dominated by $D_{sp}'\setminus \{v_1,v_2, \ldots, v_7\}$.

A vertex $u \in D_{sp}'\setminus \{v_1,v_2, \ldots, v_7\}$ can only be semipaired with a vertex in the set $\{v_1, v_2,v_5,v_6\}$ in $D_{sp}'$. If $\phi(v) = 3$, then only one vertex in the set $\{v_1, v_2,v_5,v_6\}$ is semipaired with a vertex $w$ such that $w \in D_{sp}' \setminus \{v_1, v_2,\ldots, v_7\}$. If $w \in V(G)$, then in $D_{sp}$, the vertex $v$ will be semipaired with vertex $w$. If $w\notin V(G)$, that is, if $w$ is obtained by splitting some vertex $u$ of degree~$4$ in $G$, then $v$ will be semipaired with vertex $u$.

Suppose that $\phi(v) \geq 4$ and more than two vertices in the set $\{v_1, v_2,v_5,v_6\}$ are semipaired in $D_{sp}' \setminus \{v_1, v_2,\ldots ,v_7\}$. For simplicity, suppose $\{x,y\} \subseteq D_{sp}' \setminus \{v_1,v_2, \ldots, v_7\}$ where $x$ and $y$ are semipaired with some vertices in the set $\{v_1, v_2,v_5,v_6\}$. Note that $V(G') \setminus \{v_1,v_2, \ldots, v_7\}$ is dominated by $(D_{sp}' \setminus \{v_1,v_2, \ldots, v_7\}) \cup \{v\}$. If  $\{x,y\} \subseteq D_{sp}$, then semipair $x$ with $v$ in $D_{sp}$. Now it may be the case that $y$ is not semipaired in $D_{sp}$. If the vertex $y \in D_{sp}$ has no partner in $D_{sp}$ and $N_G(y) \subseteq D_{sp}$, then remove $y$ from $D_{sp}$; otherwise, we will include another vertex $u \notin D_{sp}$ which is at distance at most~$2$ from $y$.

The resulting set $D_{sp}$ is a semi-PD-set of $G$. We note that, $|D_{sp}| \leq |D_{sp}'|-2k$. Thus, $|D_{sp}^{*}| \leq |D_{sp}'^{*}|-2k$ and hence, $|D_{sp}'^{*}| = |D_{sp}^{*}|+2k$. Consequently, we have $|D_{sp}|-|D_{sp}*| \leq |D_{sp}'|-|D_{sp}'^{*}|$.
\end{proof}

Finally, since $G$ is a graph with maximum degree $4$, for any dominating set $D$ of $G$, we have $|D| \geq |V(G)|/5$. In particular, $|D_{sp}^*| \geq |V(G)|/5$. Since $k\leq |V(G)|\leq 5|D_{sp}^*|$, we have $|D_{sp}'^{*}| \leq |D_{sp}^{*}|+2k \leq 11|D_{sp}^*|$. Hence, we may conclude that $h$ is an L-reduction from the \textsc{Minimum Semipaired Domination} for graphs with maximum degree $4$ to the \textsc{Minimum Semipaired Domination} for graphs with maximum degree $3$ with $\alpha = 11$ and $\beta = 1$.
\end{proof}

\section{Conclusion}
The \textsc{Semipaired Domination Decision} problem is already known to be NP-complete for chordal graphs. In this paper, we study this problem for two important subclasses of chordal graphs: split graphs and block graphs. We show that the \textsc{Semipaired Domination Decision} problem remains NP-complete for split graphs, and we propose a linear-time algorithm to compute a minimum cardinality semi-PD-set of a block graph. We also prove that the \textsc{Minimum Semipaired Domination} is APX-complete for graphs with maximum degree $3$. It will be interesting to study the complexity status of the problem for other important subclasses of chordal graphs, for example strongly chordal graphs, doubly chordal graphs, undirected path graphs etc.

\nocite{*}
\bibliographystyle{abbrvnat}
\bibliography{Semi}
\label{sec:biblio}

\end{document}